\documentclass[11pt]{article}
\usepackage{amsthm}
\usepackage{amsmath}
\usepackage{amssymb}
\usepackage{refcount}
\usepackage{fullpage}
\usepackage{paralist}
\usepackage{appendix}
\usepackage{graphicx}
\usepackage{bm}
\usepackage{natbib}
\usepackage[ruled,boxed,vlined]{algorithm2e}
\DontPrintSemicolon
\usepackage{verbatim}

\renewcommand{\phi}{\varphi}

\newcommand{\hide}[1]{ }

\renewcommand{\mathbf}{\bm}

\theoremstyle{plain}
\newtheorem{theorem}{Theorem}[section]

\newtheorem{proposition}[theorem]{Proposition}
\newtheorem{lemma}[theorem]{Lemma}
\newtheorem{fact}[theorem]{Fact}

\newtheorem{definition}{Definition}
\newtheorem*{remark}{Remark}
\newtheorem{observation}{Observation}
\renewcommand{\include}{\input}

\usepackage{tikz}
\usetikzlibrary{decorations.pathreplacing,angles,quotes}
\usetikzlibrary{arrows}
\usepackage{float}
\usepackage{hyperref}
\usepackage{xcolor}
\usepackage{url}
\hypersetup{colorlinks,linkcolor=blue,citecolor=blue,urlcolor=blue}

\usetikzlibrary{patterns}

\newcommand{\constrevtree}{T^R} 
\newcommand{\constdistree}{T^D} 
\newcommand{\optrevtree}{T^O}
\newcommand{\optdistree}{T^O}
\newcommand*{\field}[1]{\mathbb{#1}}%
\usepackage[noend]{algorithmic}
\usepackage{subcaption}

\title{Hierarchical Clustering via Sketches and Hierarchical Correlation Clustering}
\author{Danny Vainstein%
	\thanks{School of Computer Science, Tel-Aviv University and Google Research. Email: dannyvainstein@gmail.com}
	\and
	Vaggos Chatziafratis
	\thanks{Google Research. Emails: \{vaggos, gcitovsky, anandbr, mahdian\}@google.com}
	\and
	Gui Citovsky \footnotemark[2]
	\and
	Anand Rajagopalan \footnotemark[2]
	\and
	Mohammad Mahdian \footnotemark[2]
	\and
	Yossi Azar
	\thanks{School of Computer Science, Tel-Aviv University. Email: azar@tau.ac.il. Research upported in part by the Israel Science Foundation (grant No. 2304/20 and grant No. 1506/16).}
}

\begin{document}

\maketitle

\begin{abstract}
Recently, Hierarchical Clustering (HC) has been considered through the lens of optimization. In particular, two maximization objectives have been defined. Moseley and Wang defined the \emph{Revenue} objective to handle similarity information given by a weighted graph on the data points (w.l.o.g., $[0,1]$ weights), while Cohen-Addad et al. defined the \emph{Dissimilarity} objective to handle dissimilarity information. In this paper, we prove structural lemmas for both objectives allowing us to convert any HC tree to a tree with constant number of internal nodes while incurring an arbitrarily small loss in each objective. Although the best-known approximations are 0.585 and 0.667 respectively, using our lemmas we obtain approximations arbitrarily close to 1, if not all weights are small (i.e., there exist constants $\epsilon, \delta$ such that the fraction of weights smaller than $\delta$, is at most $1 - \epsilon$); such instances encompass many metric-based similarity instances, thereby improving upon prior work. Finally, we introduce Hierarchical Correlation Clustering (HCC) to handle instances that contain similarity and dissimilarity information simultaneously. For HCC, we provide an approximation of 0.4767 and for complementary similarity/dissimilarity weights (analogous to $+/-$ correlation clustering), we again present nearly-optimal approximations.
\end{abstract}

\section{INTRODUCTION}

Clustering is a fundamental problem in unsupervised learning and has been widely and intensively explored. Classically, one considers a set of data points (with some notion of either similarity or dissimilarity between every pair) and then partitions these data points into sets. In order to differentiate between different partitions, many classical \textit{flat} clustering objectives have been introduced, such as $k$-means, $k$-median and $k$-center. However, what if one would like a more granular view of the clusters (specifically, to understand the relations between data points within a given cluster)?

To explore these questions, the notion of \textit{Hierarchical Clustering} (HC) has been introduced. One way of studying this notion is through the lens of optimization. \citet{a_cost_function_for_similarity-based_hierarchical_clustering} initiated this line of work, inspiring others to consider several different objectives. Two notable objectives that we will consider in our paper are the Revenue and Dissimilarity objectives. 

The problem is defined as follows. We are given a set of data points with some notion of similarity (or dissimilarity) between every pair of points which is defined by a weighted graph, $G = (V,E,w)$ such that $V$ is our set of data points, $|V|=n$ and $w: E \rightarrow \mathbb{R}_{\geq 0}$. We then define an HC tree as a rooted tree with leaves in bijective correspondence with the original data points. Intuitively, we would expect a "good" HC tree $T$ to split more similar data points towards the leaves of the tree. When we are given similarity weights, this corresponds to larger weights. Thus, \citet{Approximation_Bounds_for_Hierarchical_Clustering:_Average_Linkage} proposed to maximize the \textit{Revenue} objective:

\begin{equation}
\tag{\texttt{Rev-HC}}
\label{equation.objective_rev}
rev_G(T) = \sum_{i<j} w_{ij} (n - |T_{ij}|),    
\end{equation}

\noindent where $T_{ij}$ is the subtree rooted at the lowest common ancestor (LCA) of $i$ and $j$, and $|T_{ij}|$ denotes the number of leaves of $T_{ij}$ for any binary tree $T$. The second objective we consider was defined within the dissimilarity realm by \citet{Hierarchical_Clustering:_Objective_Functions_and_Algorithms}. In this case, larger weights corresponds to dissimilar data points. Therefore, a (binary) tree $T$ should be rewarded for splitting larger weights towards its root and thus their \textit{Dissimilarity} objective is to maximize:
\begin{equation}
\tag{\texttt{Dis-HC}}
\label{equation.objective_dis}
dis_G(T) = \sum_{i<j} w_{ij} |T_{ij}|.
\end{equation}
Note that when considering both objectives, we may (and will) assume w.l.o.g. that $w_{ij} \in [0,1]$. 

Since the objectives have been introduced, there has been a line of work designing approximation algorithms. For the \ref{equation.objective_rev} objective, the best approximation ratio is 0.585~\citep{Hierarchical_Clustering:_a_0.585_Revenue_Approximation}, while for the \ref{equation.objective_dis} the best ratio is 0.667~\citep{Hierarchical_Clustering_better_than_Average_Linkage}. In terms of hardness, both problems have been proven to be APX-hard \citep{Bisect_and_Conquer:_Hierarchical_Clustering_via_Max-Uncut_Bisection,chatziafratisinapproximability} and thus do not admit optimal or even arbitrarily close to optimal approximations. Given these results, it seems natural to ask whether this hardness is inherent in the objectives, or rather can be somehow circumvented. Towards that end, we consider the following question:

\begin{center}
    \emph{Is there a large class of interesting instances that can be shown to have significantly better approximations?}
\end{center} 

Surprisingly, we show that if we consider instances with weights that are not all small (see Definition \ref{definition.rho_tau_weighted}) then the above holds true. First, we obtain approximations arbitrarily close to optimal (specifically, Efficient Polynomial Time Randomized Approximation Schemes (Efficient-PRAS)) for both \ref{equation.objective_rev} and \ref{equation.objective_dis} objectives. Interestingly, in order to do so we first consider a tree's \textit{sketch} (defined as the tree resulting from removing all its leaves (and corresponding edges)). Even though it is well known that the optimal trees for these settings are binary (and therefore contain $n - 1 = \Omega(n)$ nodes), we show that there exist trees with constant sized (i.e., a constant number of nodes and edges) sketch, for both objectives, that approximate the optimal values arbitrarily good. \textbf{We stress that this holds true for any HC instance, and not only if not all input weights are small.} We then leverage the seminal work of \citet{Property_Testing_and_its_Connection_to_Learning_and_Approximation} in order to obtain approximations arbitrarily close to optimal, if not all weights are small. 

Second, we show that many interesting, and formerly researched problems, are encapsulated by these types of instances. Specifically, we show that a large family of metric-based similarity instances (as defined by \citet{Hierarchical_Clustering_for_Euclidean_Data} - see Subsection \ref{subsection.metric}) are such instances, and thus admit approximations arbitrarily close to optimal. We note that this partially answers an open question raised in their work of whether there exist good approximation algorithms for low dimensions. We also note that our results immediately provide an Efficient-PRAS for similarity instances defined by a Gaussian Kernel in high dimensions when the minimal similarity is $\delta = \Omega(1)$ which was specifically considered by \citet{Hierarchical_Clustering_for_Euclidean_Data}; improving the approximation from  $\frac{1+\delta}{3}$ to an approximation that is arbitrarily close to optimal. Finally, we show that these results also provide an approximation that is arbitrarily close to optimal, for the +/- Hierarchical Correlation Clustering problem (defined next).

Up until now we have only considered instances handling either similarity or dissimilarity information, \emph{but not both}. In many scenarios, however, both types of information are accessible simultaneously. These scenarios have been tackled within the realm of correlation clustering both in theory (e.g., \citet{Correlation_Clustering,Correlation_Clustering:_maximizing_agreements_via_semidefinite_programming,charikar2005clustering,ailon2008aggregating,chawla2015near}) and in practice (e.g., \citet{bonchi2014correlation,cohen2001learning}). However, this line of work has been centered around flat clustering. With that in mind, it is natural to ask:

\begin{center}
    \emph{In presence of mixed information, how can we extend the notion of Correlation Clustering to hierarchies?}
\end{center}

In order to answer the question, we introduce the Hierarchical Correlation Clustering objective. The objective interpolates naturally between the \ref{equation.objective_rev} and \ref{equation.objective_dis} objectives. Again, we are given a set of data points; however, in this case every pair of data points $i$ and $j$ are given a similarity weight $w^s_{ij}$ and a dissimilarity weight $w^d_{ij}$. The objective is then defined as,
\begin{equation}
\tag{\texttt{HCC}}
\label{equation.objective_hcc}
hcc_G(T) = \sum_{i<j} w^s_{ij} (n - |T_{ij}|) + \sum_{i<j} w^d_{ij} |T_{ij}|.
\end{equation}

Observe that this objective is a direct generalization of the \ref{equation.objective_rev} and \ref{equation.objective_dis} objectives simply by letting either $w^d_{ij}=0$ or $w^s_{ij}=0$ respectively. Moreover, it captures the fact that similar points (i.e., large $w^s_{ij}$) should be separated towards the tree's leaves (yielding a large $n-|T_{ij}|$ coefficient), whereas dissimilar points (i.e., large $w^d_{ij}$) should be split towards the tree's root (yielding a large $|T_{ij}|$ coefficient). 

Finally, we consider the $+/-$ variant of correlation clustering \citep{Correlation_Clustering} extended to hierarchies as well. We define this objective as the \ref{equation.objective_hcc} objective reduced to instances that guarantee $w^s_{ij} = 1 - w^d_{ij}$ for all data points $i$ and $j$. We will refer to this objective as the $\texttt{HCC}^{\pm}$ objective. This may be motivated by the following folklore example: assume one is given a document classifier $f$ that returns a confidence level in $[0,1]$ corresponding to how certain it is that two documents are similar. Thus, 1 minus the confidence level may be seen as how confident the classifier is that the two documents are dissimilar. For further comments regarding our formulation and how it is related to the correlation clustering objectives of \citet{Correlation_Clustering} and of \citet{Correlation_Clustering:_maximizing_agreements_via_semidefinite_programming}, see Section \ref{section.hcc}.

\noindent \textbf{Contributions of this paper.} With respect to the \ref{equation.objective_rev} and \ref{equation.objective_dis}  objectives:
\begin{itemize}
\item We present structural lemmas for the revenue and dissimilarity settings that provide a way of converting optimal trees in both settings such that the resulting trees (1) are of constant sketch size and (2) approximate the respective objectives arbitrarily close (see Figure \ref{figure.general_reduction_short} for an example). Note that this result holds for \textbf{any} similarity/dissimilarity input graphs. 

\item We use the resulting trees in order to obtain Efficient-PRAS's for revenue or dissimilarity instances with not all small weights (see Definition \ref{definition.rho_tau_weighted}). We note that this includes an Efficient-PRAS for any similarity Guassian Kernel based  instances with minimal weight $\delta = \Omega(1)$ (specifically considered by \citet{Hierarchical_Clustering_for_Euclidean_Data}).

\item We show that many metric-based similarity instances in fact do not have all small weights, thus admitting Efficient-PRAS's. We note that this partially solves the case where the metric's dimension is constant (raised in \cite{Hierarchical_Clustering_for_Euclidean_Data}).
\end{itemize}

With respect to the \ref{equation.objective_hcc} objective:
\begin{itemize}
    \item We present a 0.4767 approximation for the \ref{equation.objective_hcc} objective by extending the proof of \cite{Hierarchical_Clustering:_a_0.585_Revenue_Approximation} to include dissimilarity weights. 
    \item We combine our Revenue and Dissimilarity algorithms to produce an Efficient-PRAS for the $\texttt{HCC}^{\pm}$ objective.
\end{itemize}

\noindent \textbf{Techniques.} In order to reduce HC trees to trees with constant sketch that approximate the \ref{equation.objective_rev} and \ref{equation.objective_dis} objectives arbitrarily closely, we use the following techniques. For both objectives the first step is to consider an optimal solution, $T$, and contract it (i.e., contract some subgraphs of $T$ into single nodes) into an intermediate tree denoted as $K(T)$. Briefly, $K(T)$ is generated by recursively finding a constant-sized set of edges whose removal creates a set of trees, each containing a small and roughly equal number of data points. Thereafter, each such tree is contracted (within $T$) to a single node. This results in $K(T)$ that guarantees that (1) it contains a constant number of nodes and (2) its structure resembles that of $T$ which allows us to easily convert it to the final revenue/dissimilarity tree. Note that during this process of contraction, some data points may have been contracted as well (see Figure \ref{figure.hc_tree_to_skeleton}). Next we describe, at a high level, how to convert $K(T)$ to a proper revenue/dissimilarity tree.

\textit{Revenue setting.} In the revenue setting we convert $K(T)$ to a tree denoted by $\constrevtree$, such that $\constrevtree$ has a constant-sized sketch and approximates the revenue gained by $T$ up to an arbitrarily small constant factor. In order to do so we replace each contracted node in $K(T)$ with a ``star" structure (which is an auxiliary node with the contracted data points connected as its children) - see Figure \ref{figure.skeleton_to_hc_tree}. Note that there is a trade-off between $\constrevtree$'s internal tree size and the revenue approximation factor guaranteed (see Section \ref{section.revenue} for formal details).

\textit{Dissimilarity setting.} In the dissimilarity setting we convert $K(T)$ to a tree denoted by $\constdistree$ such that $\constdistree$ has a constant-sized sketch and approximates the dissimilarity gained by $T$ up to an arbitrarily small constant factor. Instead of replacing the contracted node with a ``star" structure as in the revenue case, we replace it with a random ``comb" structure (formally defined in Section \ref{section.dissimilarity} and depicted in Figure \ref{figure.skeleton_to_hc_tree}). Also here, there exists a trade-off between $\constdistree$'s size and the approximation factor.

\noindent \textbf{Related Work.} HC has been extensively studied and therefore many variations have been considered (for a survey on the subject, see \citet{A_Survey_of_Clustering_Data_Mining_Techniques}). The work on HC trees began within the realm of phylogenetics \citep{Numerical_taxonomy, A_model_for_taxonomy} but has since then expanded to many other domains (e.g., genetics, data analysis and text analysis - \citet{Broad_patterns_of_gene_expression_revealed_by_clustering_analysis_of_tumor_and_normal_colon_tissues_probed_by_oligonucleotide_arrays, Class_Based_n_gram_Models_of_Natural_Language, Interactively_Exploring_Hierarchical_Clustering_Results}).

As stated earlier, Dasgupta elegantly linked the fields of approximation algorithms and HC trees, thereby initiating this line of work. Formally, given an HC tree, $T$, \citet{a_cost_function_for_similarity-based_hierarchical_clustering} considered the problem of minimizing its cost, $cost_G(T) = \sum w_{ij}|T_{ij}|$. In his work, Dasgupta showed that recursively finding a sparsest cut results in a $O(\log^{1.5}n)$ approximation. This analysis was later improved to $O(\sqrt{\log n})$ \citep{Approximate_Hierarchical_Clustering_via_Sparsest_Cut_and_Spreading_Metrics, Hierarchical_Clustering:_Objective_Functions_and_Algorithms}. \citet{Approximate_Hierarchical_Clustering_via_Sparsest_Cut_and_Spreading_Metrics} also showed that no constant approximation exists (assuming the Small Set Expansion hypothesis).

Later, \citet{Approximation_Bounds_for_Hierarchical_Clustering:_Average_Linkage} considered the \ref{equation.objective_rev} objective (defined earlier). \citet{Hierarchical_Clustering_better_than_Average_Linkage} showed a $0.3364$ approximation through the use of semi-definite programming. Later, \citet{Bisect_and_Conquer:_Hierarchical_Clustering_via_Max-Uncut_Bisection} made use of the \textsc{Max-Uncut Bisection} problem in order to prove a $0.4246$ approximation. Finally, \citet{Hierarchical_Clustering:_a_0.585_Revenue_Approximation} improved upon this by showing a $0.585$ approximation, by proving the existence of a bisection which yields large revenue.

\cite{Hierarchical_Clustering:_Objective_Functions_and_Algorithms} considered the \ref{equation.objective_dis} objective (defined earlier). In their work they showed that the Average-Linkage algorithm is a $\frac{1}{2}$ approximation and then improved upon this by presenting a simple algorithm achieving a $\frac{2}{3}$ approximation. \cite{Hierarchical_Clustering_better_than_Average_Linkage} then showed a further improvement by presenting a more intricate algorithm that achieves a $0.6671$ approximation.



Since the work of \citet{Correlation_Clustering}, correlation clustering has been extensively studied. Considering more theoretical settings, the work most relevant to ours is that of \cite{Correlation_Clustering:_maximizing_agreements_via_semidefinite_programming}, showing a 0.766-approximation for a maximization version of the problem, interpolating between roundings from multiple hyperplanes, instead of just one as in \cite{Improved_Approximation_Algorithms_for_Maximum_Cut_and_Satisfiability_Problems_Using_Semidefinite_Programming}. The problem is also highly significant in practice as well - see e.g., spam filtering \citep{ramachandran2007filtering}, image segmentation \citep{kim2011higher} and co-reference resolution \citep{cohen2002learning, elmagarmid2006duplicate}.

\begin{figure}[H]
\centering
\includegraphics[scale=0.45]{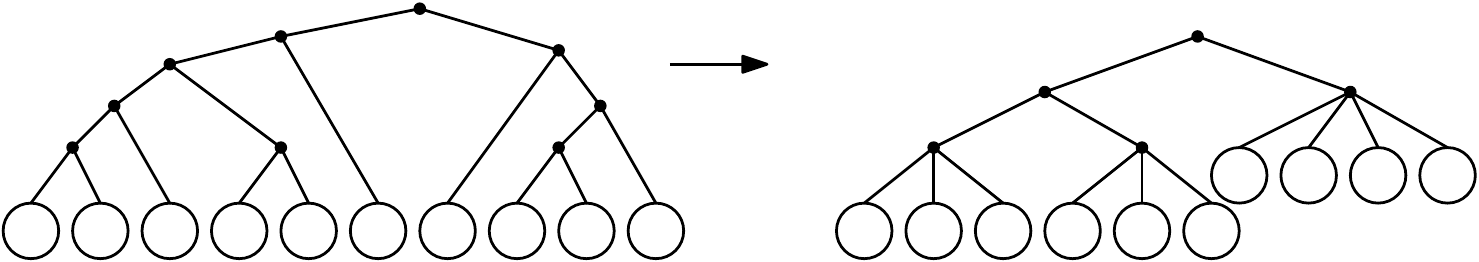}
\caption{Converting an HC tree to a tree of constant Sketch while approximating the goal function.}
\label{figure.general_reduction_short}
\end{figure}

\begin{figure}[H]
\centering
\includegraphics[scale=0.5]{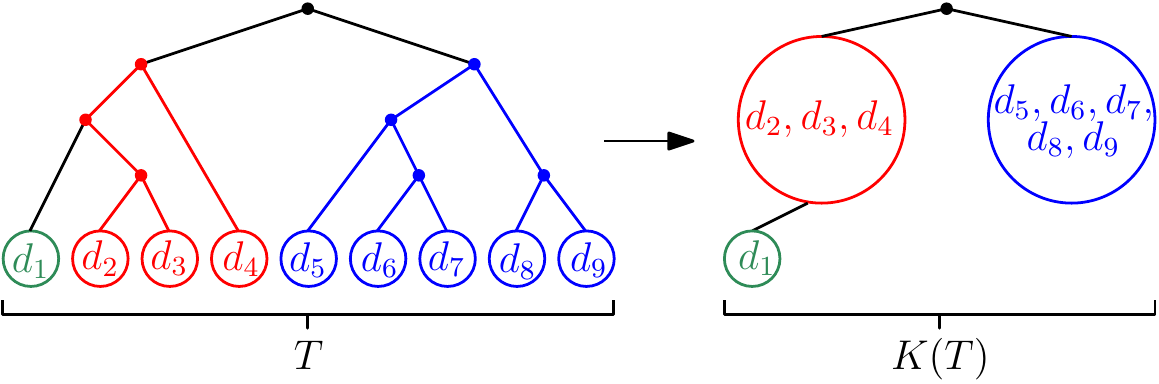}
\caption{Converting an HC tree $T$ to $K(T)$.}
\label{figure.hc_tree_to_skeleton}
\end{figure}

\begin{figure}[H]
\centering
\includegraphics[scale=0.5]{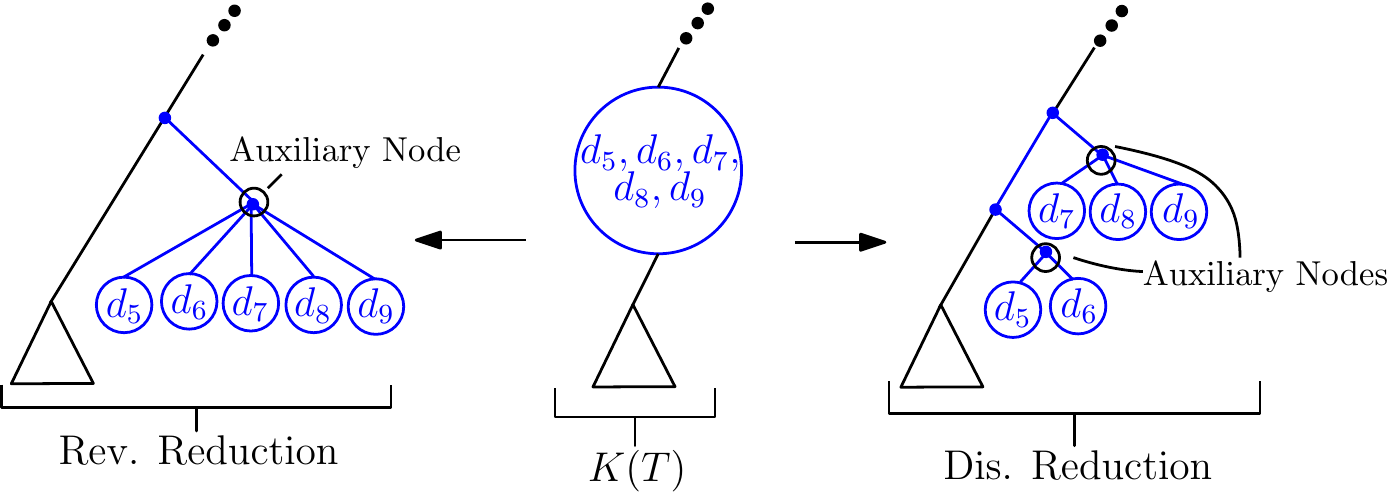}
\caption{Converting $K(T)$ to an HC tree for each goal function.}
\label{figure.skeleton_to_hc_tree}
\end{figure}


\section{PRELIMINARIES}
\label{section.notations_and_preliminaries}

We first consider several graph-specific definitions.

\begin{definition}
Given a tree $T$ and a set of edges $F \subset E(T)$, let $T-F$ denote the set of trees that results from removing $F$ from $E(T)$. Furthermore, given a set of nodes $U \subset V(T)$, let $T-U$ denote the set of trees that results from removing $U$ (and any edge that has a node in $U$) from $T$.
\end{definition}

\begin{definition}
\label{definition.contraction}
Given a graph $G$ and a subset of edges $U \subset V(G)$ we define the contraction of $U$ as the replacement of $U$ within $G$ with a single node attached to all edges which were formerly attached to $U$.
\end{definition}

As pointed out by \citet{Hierarchical_Clustering_better_than_Average_Linkage}, the average-linkage algorithm generates $\frac{(n-2)}{3} \sum w_{ij}$ revenue and $\frac{2(n-2)}{3} \sum w_{ij}$ dissimilarity, yielding the following facts: 

\begin{fact}
\label{fact.random_alg_yield_rev_third}
$rev(\optrevtree) \geq \frac{(n-2)}{3} \sum_{i<j}w_{ij}$, where $\optrevtree$ denotes the optimal revenue tree.
\end{fact}

\begin{fact}
\label{fact.random_alg_dissimilarity}
$dis(\optrevtree) \geq \frac{2n}{3}\sum_{i<j} w_{ij}$, where $\optdistree$ denotes the optimal dissimilarity tree.
\end{fact}

Furthermore, as pointed out by \citet{a_cost_function_for_similarity-based_hierarchical_clustering} all binary trees generate the same dissimilarity on instances defined by cliques (i.e., $w_{ij}=1$ for all $i$ and $j$).

\begin{fact}
\label{fact.revenue_of_clique}
$\sum_{i,j} |T_{ij}| = \frac{2n}{3} {n \choose 2}$. 
\end{fact}

\noindent \textbf{A note on non-binary HC trees.} Even though the \ref{equation.objective_rev} and \ref{equation.objective_dis} objectives are defined for binary trees, we make use of star structures. A star structure is simply a node that contains more than two data points as children (and therefore leaves). We use these star structures as a proxy for any binary tree containing the same set of data points. More formally, by replacing the star structure (within some larger tree) with any binary tree containing the same set of data points and then rooting it in the same place within the original tree, the goal function would only increase. 

In the revenue case this follows immediately. In the dissimilarity case, however, by following the definition of $T_{ij}$ plainly, clearly attaching all data points to a single root results in an optimal tree. Therefore, we instead extend the dissimilarity definition to non-binary trees as follows. Given an HC tree $T$ and internal node $v$, let $|T_v|$ denote the set of data points contained within the subtree rooted at $v$ (in particular, for any 2 data points $i$ and $j$, $|T_{ij}| = |T_{lca(ij)}|)$. We then define the dissimilarity as
\begin{equation*}
dis_G(T) = \sum w_{ij} (|T_{v_i}| + |T_{v_j}|),    
\end{equation*}
where $v_i$ and $v_j$ denote $lca(i,j)$'s children containing $i$ and $j$ in their subtree. We emphasize the fact that for binary HC trees, this definition coincides with the classic dissimilarity (since $|T_{v_i}| + |T_{v_j}| = |T_{ij}|$). Clearly any non-binary node may be replaced with a binary subgraph within the HC tree thereby only increasing the dissimilarity generated. Therefore, any of our algorithmic results apply to the binary setting (by performing these replacements). Further, all of our approximation results are with respect to optimal binary trees and thus directly apply to the binary setting.



Finally, we will use the following definitions throughout the paper. (Recall that w.l.o.g. we may assume that all weights are in $[0,1]$).

\begin{definition}
\label{definition.rho_tau_weighted}
An HC instance is said to have not all small weights if there exist constants (with respect to $|V|$) $\rho, \tau$  such that the fraction of weights smaller than $\tau$, is at most $1 - \rho$.
\end{definition}


\begin{definition}
\label{definition.EPRAS}
An algorithm is considered an Efficient-PRAS if for any $\epsilon>0$ the algorithm runs in time $f(1/\epsilon)n^{O(1)}$ and approximates the optimal solution's value up to a factor of $1 - \epsilon$ with high probability.
\end{definition}


\section{THE REVENUE CASE}
\label{section.revenue}

In this section we consider the \ref{equation.objective_rev} objective. In Subsection \ref{subsection.rev_sketch_reduction} we show how to create a tree with constant sized sketch which approximates the optimal revenue tree up to an arbitrarily small factor (for an overview see Techniques). Note that this result holds for \textbf{any} revenue instance and thus may be of independent interest. We then leverage this and in Subsection \ref{subsection.dense_rev} we present an Efficient-PRAS for instances with not all small weights. Finally, in Subsection \ref{subsection.metric} we show that a large family of metric-based similarity instances have weights that are not all small - thereby admitting Efficient-PRAS's. We note that this partially solves an open question raised by \cite{Hierarchical_Clustering_for_Euclidean_Data} regarding constant dimension instances and immediately provides Efficient-PRAS's for similarity instances defined by a Gaussian Kernel in high dimensions when the minimal similarity is $\delta = \Omega(1)$ which was specifically in their work as well.

\subsection{A Reduction to Constant Sketches}
\label{subsection.rev_sketch_reduction}

We begin by first proving the existence of a tree with constant-sized sketch that approximates the optimal tree arbitrarily well.

\begin{theorem}
\label{theorem.approximate_rev_with_constant_size}
Let $\optrevtree$ denote the optimal revenue tree and assume it contains $n$ leaves (i.e., data points). Then, for any $\epsilon>0$, there exists a tree $\constrevtree$ such that (i) $\constrevtree$ contains $\Theta(1/\epsilon)$ internal nodes each with at most  $3\epsilon n$ children, and (ii) $rev(\constrevtree) \geq (1 - 19 \epsilon)rev(\optrevtree)$.
\end{theorem}

In order to construct $\constrevtree$ we use a two step process: we first create an intermediate tree, denoted as $K(T)$ (to be defined) and then convert that to our final tree. In fact, this process may be applied to any binary tree $T$ (in particular, we will apply it to $\optrevtree$). Before we can define the process that generates $K(\optrevtree)$, we must first present several definitions and lemmas, the first of which was shown by \citet{a_cost_function_for_similarity-based_hierarchical_clustering} (this was not explicitly proven, and therefore we add the proof in the Appendix for completeness).

\begin{lemma}
\label{lemma.dasguptas_lemma}
Given a rooted binary tree $T$ with $n$ data points as leaves, there exists an edge whose removal creates two binary trees each with at least $\frac{n}{3}$ data points (and therefore at most $\frac{2n}{3}$). Furthermore this edge can be found in polytime.
\end{lemma}

\begin{lemma}
\label{lemma.create_balanced_forest}
Given a rooted binary tree $T$ with $n$ data points, there exists a set of edges $F$ such that $\frac{1}{3 \epsilon} \leq |F|+1 \leq \frac{1}{\epsilon}$ and the number of data points in each tree of $T - F$ is at least $\epsilon n$ and at most $3 \epsilon n$. Furthermore $F$ can be found in polytime.
\end{lemma} 

\begin{proof}[Proof of Lemma \ref{lemma.create_balanced_forest}]
Let $n$ denote the number of data points in $T$. We define the following recursive algorithm: for any binary tree instance $T$ find the edge given by Lemma \ref{lemma.dasguptas_lemma}. Remove said edge and continue recursively on both resulting trees. Stop once the input tree has less than $3 \epsilon n$ data points.

The algorithm is clearly polynomial. Let $F$ denote the set of resulting edges. Due to our stopping condition, every tree in $T-F$ contains between $\epsilon n$ and $3\epsilon n$ data points. Therefore, $\frac{1}{4 \epsilon} + 1 \leq |F| \leq \frac{1}{\epsilon}$ for $\epsilon < 1/12$.
\end{proof}

\noindent The following is a straightforward but useful lemma.

\begin{lemma}
\label{lemma.degree_3_simple_bound}
For an arbitrary tree $T$, let $V_3$ denote the set of vertices with degree $\geq 3$ and $L$ denote its set of leaves. Then, $|V_3| \leq |L| - 1$.
\end{lemma}

\begin{proof}
Let $T$ be some tree on $n$ nodes and let $\ell$ denote some leaf. We prove by induction on $n$. If $n=1$ or $n=2$ clearly we are done. Otherwise, traverse $T$ starting at $\ell$ (i.e., hopping from a node to one of its untravelled neighbours). If during this traversal we arrive at a leaf before we arrive at a node with degree $\geq 3$, then $|V_3| = 0$ and we are done. Otherwise let $u$ denote the first node we traverse with degree $\geq 3$. Remove all nodes in the traversal upto but not including $u$, denote the new tree as $T'$. 

Thus, $|V_3| \leq |V_3'| + 1$ and $|L| - 1 = |L'|$. Furthermore, since $T'$ has at most $n-1$ nodes we may use our induction hypothesis. Therefore,
\[
|V_3| \leq |V'_3| + 1 \leq |L'| = |L| - 1.
\]
\end{proof}

\begin{definition}
\label{definition.green_blue_nodes}
Given $F$ as defined by Lemma \ref{lemma.create_balanced_forest} we define two sets of nodes: blue and green, denoted by $\mathbf{B}$ and $\mathbf{G}$. A \textbf{blue node} is any node connected to any edge of $F$ or that is $T$'s root. A \textbf{green node} is any node that is not blue and that has two children, each of which contains a blue node as its descendant. 
\end{definition}

\noindent Next we define the process that given a binary tree, contracts it compactly. Given an input $T$, we denote the process' output as $K(T)$, formally defined by Algorithm \ref{algorithm.K(T)}. (See Figure \ref{figure.hc_tree_to_skeleton} for a pictorial example). We note that each contracted node might have originally contained data points. We therefore associate every contracted node, $c$ with its set of data points, $D_c$. Finally, we define the process that given any binary tree $T$, outputs $\constrevtree$ - formally defined by Algorithm \ref{algorithm.convert_T_to_\constrevtree}.\\

\begin{algorithm}[H]
    \centering
    \caption{Algorithm to convert $T$ to $K(T)$.}
    \begin{algorithmic}
        \STATE Obtain $F$ as described in Lemma \ref{lemma.create_balanced_forest}.\\
        \STATE Color the nodes green or blue as in Definition \ref{definition.green_blue_nodes}.\\
        \FOR{every tree $T_i$ in $T - (B \cup G)$}
          \STATE Contract $T_i$.
        \ENDFOR
        \STATE Return the resulting tree as $K(T)$.
    \end{algorithmic}
    \label{algorithm.K(T)}
\end{algorithm}

\begin{algorithm}[H]
    \centering
    \caption{Algorithm to convert $T$ to $\constrevtree$.}
    \begin{algorithmic}
        \STATE $K(T) \leftarrow$ Algorithm \ref{algorithm.K(T)} applied to $T$.
        \FOR{each node $c \in K(T)$ and its set of data points $D_c$}
          \STATE Attach a (new) auxiliary node as $c$'s child (in $K(T)$).
          \STATE Attach $D_c$ as the auxiliary node's children.
        \ENDFOR
        \STATE Return the resulting tree as $\constrevtree$.
    \end{algorithmic}
    \label{algorithm.convert_T_to_\constrevtree}
\end{algorithm}

\begin{remark}
\label{remark.\constrevtree_remains_binary}
We note that $\constrevtree$ remains binary (except the auxiliary nodes). This is in fact true since otherwise this internal node would have contained at least 2 children which are colored green/blue (since it may only have a single auxiliary node). Thus, there would have been a green node contained within this contracted component in contradiction to the definition of $K(T)$.
\end{remark}

\noindent In what follows we show that for any binary tree $T$, (1) $\constrevtree$ has a constant sketch and (2) $|\constrevtree_{ij}|$ is (approximately) upper bounded for any data points $i$ and $j$ (which in turn guarantees that $rev(\constrevtree)$ is close to $\optrevtree$ when $T = \optrevtree$).

\begin{lemma}
\label{lemma.\constrevtree_has_constant_size}
$\constrevtree$ contains $\Theta(1/\epsilon)$ internal nodes each with at most $3\epsilon n$ children.
\end{lemma}

\begin{proof}
We first note that a node is a leaf in $\constrevtree$ if and only if it was a leaf in $T$ (since every contracted connected component either contained data points or will have a child following the contraction). Next, we categorize the internal nodes of $\constrevtree$. These nodes are either colored (green or blue), or they are a contracted node or an auxiliary node. We denote the set of each such nodes by $G$,$B$,$C$ and $A$ respectively.

It is not hard to see that the second part of our lemma holds. This is due to the fact that by Remark \ref{remark.\constrevtree_remains_binary} every node in $G$,$B$ and $C$ has at most 2 immediate children. For nodes in $A$, by Lemma \ref{lemma.create_balanced_forest} and by $A$'s definition, we are guaranteed that any such node has at most $3\epsilon n$ children.

In order to show the first part of the lemma we bound each of the four sets of nodes. By the definition of $B$, $|B| \leq 2/\epsilon$. By definition of $A$, $|A| \leq |C|$. Furthermore, every node in $C$ has a parent that is colored green or blue and thus due to Remark \ref{remark.\constrevtree_remains_binary}, $|C| \leq 2(|G| + |B|)$. Therefore, $|A| + |C| \leq 4(|G| + |B|)$.

Next we bound $|G|$. In order to do so, we first simplify $\constrevtree$ in a way that does not affect $|G|$. Since no auxiliary node contains green nodes in their subtree, we may detach them without affecting any green or blue nodes. Furthermore, this removal upholds the fact that any green node's degree is at least 3 (since we did not remove any blue nodes). We then also remove any contracted node which now happens to be a leaf (since they too, do not affect the green or blue nodes). 

Therefore, in the resulting tree, any leaf must be blue and any green node must have degree at least 3. Thus, if we denote by $V_3$ the set of vertices with degree $\geq 3$ and by $L$ the set of leaves, then,
\[
|G| \leq |V_3| \leq |L|-1 \leq |B|-1,
\]
where the second inequality is due to Lemma \ref{lemma.degree_3_simple_bound}. Thus,
\[
|A| + |C| + |G| + |B| \leq 5(|G| + |B|) \leq 10|B| \leq 20/\epsilon.
\]

Now, in order to show the complement (i.e., $\constrevtree$ contains $\Omega(1/\epsilon)$ internal nodes) it is enough to consider Lemma \ref{lemma.create_balanced_forest} thereby concluding the proof.
\end{proof}

\begin{lemma}
\label{lemma.\constrevtree_maintains_lca}
For any two data points $i$ and $j$, $|\constrevtree_{ij}| \leq |T_{ij}| + 6 \epsilon n$.
\end{lemma}

\begin{proof}
Consider any three data points in $T$, $i,j$ and $k$, such that $k \not \in T_{ij}$. We will show that $k \not \in \constrevtree_{ij}$ for all but $6\epsilon n$ such $k$'s. In order to prove our lemma we first introduce the following notations. First, for any node $u$ we denote the set of data points contained in its induced subtree as $L(u)$. Secondly we note that any node colored green or blue in $T$ will not be contracted and therefore will appear in $V(\constrevtree)$. Finally, we observe the following given our contraction process.

\begin{observation}
\label{observation.unchanged_children_sets}
Let $v \in V(T)$ denote a child of a green/blue node and let $v^* \in V(\constrevtree)$ denote the node that contracted $v$ in $\constrevtree$. Therefore, $L(v) = L(v^*)$.
\end{observation}

\begin{observation}
\label{observation.pair_under_same_auxiliary_parent}
Data points $i$ and $j$ appear under the same auxiliary node in $\constrevtree$ if and only if $i$ and $j$ were contained in the same tree of $T - (B \cup G)$.
\end{observation}


Recall that our goal is to show that if $k \not \in T_{ij}$ then $k \not \in \constrevtree_{ij}$. Towards that end, denote by $v_{ij}$ (resp. $v_{ik}$ and $v_{jk}$) $i$ and $j$'s LCA in $T$. Therefore, $v_{ik} = v_{jk}$ and $v_{ij}$ is a descendant of $v_{ik}$. Furthermore, let $\{T^{B \cup G}_\ell \}$ denote the set of trees defined by $T-(B \cup G)$ and let $T^{B \cup G}_i$ (resp. $T^{B \cup G}_j$ and $T^{B \cup G}_k$) denote the tree in $T-(B \cup G)$ containing $i$ (resp. $j$ and $k$). 

We first assume $k \not \in T^{B \cup G}_i$ and $k \not \in T^{B \cup G}_j$. Therefore, a green or blue node must be either on the path $k \rightarrow v_{ik}$, or on the path $v_{ij} \rightarrow v_{ik}$. Otherwise there must be a green or blue node on the path $i \rightarrow v_{ij}$ and on the path $j \rightarrow v_{ij}$. We consider each case separately. (See Figure \ref{figure.proof_of_structural_lemma}).




\begin{figure}[H]
\centering
\includegraphics[scale=0.7]{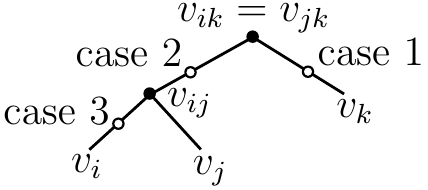}
\caption{Explanation to proof of Lemma \ref{lemma.\constrevtree_maintains_lca} (such that $v_a = a$ for $a \in \{i,j,k\}$).}
\label{figure.proof_of_structural_lemma}
\end{figure}
\textit{Case 1.} There exists a blue or green node on the path $k \rightarrow v_{ij}$: We further split this case into two cases. The first is that $i$ and $j$ are part of the same tree of $T-(B \cup G)$. In this case they will end up under the same auxiliary node and due to Observation \ref{observation.pair_under_same_auxiliary_parent} we are guaranteed that $k \not \in \constrevtree_{ij}$. The second case is that $i$ and $j$ are not part of the same tree and therefore there exists a blue/green node on the path $i \rightarrow j$. Thus, the node $v_{ik}$ must be green or blue and due to Observation \ref{observation.unchanged_children_sets}, $i$ and $j$'s lca will remain lower than $i$ and $k$'s in $\constrevtree$. Therefore, $k \not \in \constrevtree_{ij}$.

\textit{Case 2.} There exists a blue or green node on the path $v_{ij} \rightarrow v_{ik}$: In this case either $v_{ik}$ is green/blue and due to Observation \ref{observation.unchanged_children_sets} we are done. Otherwise some other node along  $v_{ij} \rightarrow v_{ik}$ is green/blue and then Observation \ref{observation.unchanged_children_sets} guarantees that $k$ will not enter the subtree defined by $i$ and $j$'s lca. Thus, in any case, $k \not \in \constrevtree_{ij}$.

\textit{Case 3.} There exists a green or blue node on the paths $i \rightarrow v_{ij}$ and $j \rightarrow v_{ij}$: If $v_{ij}$ is green/blue then Observation \ref{observation.unchanged_children_sets} guarantees that $k$ will not enter the subtree defined by $i$ and $j$'s lca. Otherwise, we are guaranteed to have two separate green/blue nodes, one on the path $i \rightarrow v_{ij}$ and one on the path $j \rightarrow v_{ij}$. Therefore, $v_{ij}$ must be green/blue. Hence, in either case, $k \not \in \constrevtree_{ij}$.

Thus, we have shown that in all 3 cases if $k \not \in T^{B \cup G}_i$ and $k \not \in T^{B \cup G}_j$ then $k \not \in \constrevtree_{ij}$. Since the number of data points within both $T^{B \cup G}_i$ and $T^{B \cup G}_j$ is at most $3 \epsilon n$ each, we get that at most $6 \epsilon n$ such $k$'s may be contained in $\constrevtree_{ij}$. Therefore, $|\constrevtree_{ij}| \leq |T_{ij}| + 6 \epsilon n$, concluding the proof.
\end{proof}

\noindent Finally, combining Lemmas \ref{lemma.\constrevtree_has_constant_size} and \ref{lemma.\constrevtree_maintains_lca} for $T = \optrevtree$ (i.e., the revenue optimal solution) with Fact \ref{fact.random_alg_yield_rev_third}, is enough to prove Theorem \ref{theorem.approximate_rev_with_constant_size}.

\begin{proof}[Proof of Theorem \ref{theorem.approximate_rev_with_constant_size}]
Lemma  \ref{lemma.\constrevtree_has_constant_size} is enough to prove the first bullet. We consider the second bullet. It is a known fact that $\optrevtree$ may be taken to be binary. Therefore, due to Lemma \ref{lemma.\constrevtree_maintains_lca} and Fact \ref{fact.random_alg_yield_rev_third}, we get,
\begin{align*}
rev(\constrevtree) &= 
\sum_{i < j} w_{ij}(n - |\constrevtree_{ij}|) \\ &\geq
\sum_{i < j} w_{ij}(n - |\optrevtree_{ij}| - 6 \epsilon n) \\ &=
rev(\optrevtree) - 6 \epsilon n \sum_{i < j} w_{ij} \\ &\geq
(1 - 19 \epsilon)rev(\optrevtree),
\end{align*}
where the last inequality is due to Fact \ref{fact.random_alg_yield_rev_third} and since $n$ is assumed to be large enough.
\end{proof}


\subsection{An Efficient-PRAS for Revenue Instances with Not All Small Weights}
\label{subsection.dense_rev}


In this section we consider the problem of finding an optimal revenue tree in instances with weights that are not all small and present an Efficient-PRAS. We show that in a sense this is the best one could hope for, and complement our result by showing that the problem is NP-Complete and thus does not admit an optimal, polynomial solution unless $P = NP$ (see Theorem \ref{theorem.dense_rev_NPC} in the Appendix). 


Let $\epsilon > 0$, let $|V| = n$ and $k = \lceil \frac{1}{\epsilon}\rceil$. Finally, let $\constrevtree_\epsilon$ denote the tree guaranteed by Theorem \ref{theorem.approximate_rev_with_constant_size} for $\epsilon$. We may define $\constrevtree_\epsilon$'s revenue as follows. For every one of $\constrevtree_\epsilon$'s internal nodes $i$, denote by $D_i$ its set of children that are data points. Furthermore, let $W_{ij}$ denote the total weight of the set of (similarity) edges crossing between $D_i$ and $D_j$. Therefore, $rev(\constrevtree_\epsilon) = \sum_{i < j} (|W_{ij}| \sum_{\ell} |D_\ell|)$, where the second summation is over all sets $D_\ell$ not contained in $\constrevtree_{ij}$ (as defined by $\constrevtree_\epsilon$'s sketch). We note that due to Theorem \ref{theorem.approximate_rev_with_constant_size}, the first summation is over at most $\Theta(k)$ entries (specifically, at most $20\cdot k$).

Next, we consider the General Partitioning Property Tester of \citet{Property_Testing_and_its_Connection_to_Learning_and_Approximation}. Given values $\alpha_i$ and $\beta_{ij}$ (representing the sizes of the data point sets and the weight of edges between every pair of sets) the property tester allows us to test whether there exists a graph partition with set sizes $\alpha_i$, and weight of edges crossing between the different sets $\beta_{ij}$. The property tester also takes as input $\epsilon_{err}$ and $\delta$ which define the error in $\alpha_i$ and $\beta_{ij}$ and the probability of failing, respectively. Formally, we denote this as $PT(\{\alpha_i\}, \{\beta_{ij}\}, \epsilon_{err}, \delta)$. Thereafter, the property tester returns the following: if there exists a partition upholding the values $\alpha_i$ and $\beta_{ij}$ then the tester returns this partition up to an additive error of $n\epsilon_{err}$ in the sizes of $\alpha_i$ and additive error of $n^2\epsilon_{err}$ in the sizes $\beta_{ij}$. If such a partition does not exist, the tester returns that such a partition does not exist.

Overall, this suggests an algorithm that guesses $\constrevtree_\epsilon$ by guessing a tree of size $20 \cdot k$ (see Theorem \ref{theorem.approximate_rev_with_constant_size}) and guessing $\alpha_i$ and $\beta_{ij}$ (simply through iteration). Unfortunately, guessing $\alpha_i$ and $\beta_{ij}$ exactly would only yield a PRAS. To obtain an Efficient-PRAS, we guess $\alpha_i$ upto a factor of $\epsilon^2$ and $\beta_{ij}$ up to a factor of $\epsilon^3$. This yields Algorithm \ref{algorithm.guess_OPT_eps_in_dense_rev_case}. Lemma \ref{lemma.algorithms_approximation} (proved in the Appendix) guarantees the approximation needed.\\

\begin{algorithm}[H]
\caption{EPRAS for Revenue case.}
\begin{algorithmic}
\STATE Enumerate over all trees, $T$, with $k$ internal leaves.
\FOR{each such $T$}
	\FOR{$\{\alpha_i\}_{i\leq k} \subset \{i \epsilon^2 n: i \in \field{N} \land  i \leq \frac{3}{\epsilon} \}$}
		\FOR{$\{\beta_{ij}\}_{i\leq k,j\leq k} \subset \{i \epsilon^3 n^2: i \in \field{N} \land  i \leq \frac{9}{\epsilon} \}$}
			\STATE Run $PT(\{\alpha_i\}, \{\beta_{ij}\}, \epsilon_{err}=\epsilon^3, \delta)$.
		\ENDFOR
	\ENDFOR
	\STATE Compute the revenue given $T$ and $PT$'s output.
\ENDFOR
\STATE Return the maximal revenue tree encountered.
\end{algorithmic}
\label{algorithm.guess_OPT_eps_in_dense_rev_case}
\end{algorithm}

\begin{lemma}
\label{lemma.algorithms_approximation}
For every $\epsilon>0$, Algorithm \ref{algorithm.guess_OPT_eps_in_dense_rev_case} guarantees an approximation factor of $(1 - 18 \epsilon - \frac{12\epsilon}{\rho\tau})$.
\end{lemma}



We note that the error from the property tester is offset by the revenue from the optimal solution. 

\begin{theorem}
\label{theorem.dense_rev_epras.epras}
Algorithm \ref{algorithm.guess_OPT_eps_in_dense_rev_case} is an Efficient-PRAS.
\end{theorem}

\begin{proof}
Lemma \ref{lemma.algorithms_approximation} guarantees that there exists $\hat{\epsilon} > 0$ (specifically, $\hat{\epsilon} = 18 \epsilon + \frac{12\epsilon}{\rho\tau}$) such that our algorithm is a $1-\hat{\epsilon}$ approximation. The property tester runs in time, $exp(\log (\frac{1}{\delta \epsilon_{err}}) (\frac{O(1)}{\epsilon_{err}})^{k+1}) + O(\frac{\log(k / (\epsilon_{err} \delta))}{\epsilon_{err}^2}) n$. Further, we call the tester $k^k \cdot (3/\epsilon)^k \cdot (9/\epsilon)^{k^2}$ times. Now, since $\epsilon < \hat{\epsilon}$, if $\epsilon_{err} = \epsilon^3$ then the algorithm is an Efficient-PRAS.
\end{proof}


\subsection{Metric-Based Similarity Instances}
\label{subsection.metric}

We follow the definitions as seen in \citet{Hierarchical_Clustering_for_Euclidean_Data}. Suppose that our data points lie on a metric $M$ with doubling dimension $D(M)$. Define a non-increasing function $g : \mathbb{R}_{\geq 0} \rightarrow [0,1]$. Given two data points $i$ and $j$ let $d_{ij}$ denote their distance as defined by our metric. Furthermore, we define the metric-based similarity weights $w_{ij} = g(d_{ij})$.

Define $A(\epsilon) = A$ to be the tree generated by the algorithm that adds a constant $\epsilon$ to all weights and then runs Algorithm \ref{algorithm.guess_OPT_eps_in_dense_rev_case} for $\rho, \tau$-weighted instances. We note that $A$ is well defined since the altered weights define a graph with not all small weights for $\tau = \epsilon$ and $\rho = 0$.

The following theorem shows that for a large class of functions $g$ and metrics $M$, algorithm $A$ is in fact an Efficient-PRAS.

\begin{theorem}
\label{theorem.revenue.metric_epras}
Assume the metric's doubling dimension guarantees $D(M) = O(1)$ and $g$ is scale invariant and $\ell$-Lipschitz continuous for $\ell = O(1)$. Then, $A$ is an Efficient-PRAS for the induced Revenue instance.
\end{theorem}

\begin{proof}
Let $w_{ij} = g(d_{ij})$ and let $w'_{ij} = w_{ij} + \epsilon$. Denote by $O$ and $O'$ the trees which generate the maximal revenue with respect to $w_{ij}$ and $w'_{ij}$ respectively. Finally, given an HC tree $T$, let $Rev(T)$ and $Rev'(T)$ denote the revenue generated by $T$ with respect to $w_{ij}$ and $w'_{ij}$ respectively. 

By Theorem \ref{theorem.dense_rev_epras.epras} we are guarnateed that for any constant $\delta > 0$, $Rev'(A) \geq (1 - \delta)Rev'(O')$. Furthermore, by the definitions of $O$ and $O'$ we have that $Rev'(O') \geq Rev'(O)$. Therefore,
\begin{align}
\label{equation.metric.1}
Rev'(A) \geq (1 - \delta)Rev'(O') \geq (1 - \delta)Rev'(O).
\end{align}

By Fact \ref{fact.revenue_of_clique} and since $w_{ij} + \epsilon = w'_{ij}$ we are guaranteed that for any tree $T$, $Rev(T) = Rev'(T) - \epsilon \frac{n}{3} {n \choose 2}$. Combining this with equation \ref{equation.metric.1} we get that,

\begin{align*}
Rev(A) &=
Rev'(A) - \epsilon \frac{n}{3} {n \choose 2} \\&\geq
(1 - \delta)Rev'(O) - \epsilon \frac{n}{3} {n \choose 2} \\&=
(1 - \delta)Rev(O) - \delta \epsilon \frac{n}{3} {n \choose 2}.
\end{align*}

Let $\alpha$ denote the diameter of the metric. Since the metric is scale invariant we may assume w.l.o.g. that $\alpha = 1$. By the definition of the doubling dimension, $D(M) = D$, there are $2^{D(\ell+1)}$ balls of radius $\frac{1}{2^{\ell+1}}$ that cover the entirety of the data. Let $x_i$ denote the number of data points that belong to the $i$'th ball but not to balls $1, \ldots, i-1$. Therefore, $\sum_{i=1}^{2^{D(\ell+1)}} x_i = n$. On the other hand by Cauchy-Schwarz inequality, $\sum_{i=1}^{2^{D(\ell+1)}} x_i^2 \geq \frac{n^2}{2^{D(\ell+1)}}$. Therefore, the number of pairs of data points within the same ball is $\sum_{i=1}^{2^{D(\ell+1)}} {x_i \choose 2} \geq \frac{n^2}{2^{D(\ell+1) + 1}} - \frac{n}{2}$. Due to the fact that pairs of points that belong to the same ball are at distance of at most $\frac{1}{2^{\ell}}$ and since similarity function $g$ is defined an non-increasing, we get that,
\begin{align}
\label{equation.metric.2}
\sum_{i,j} w_{ij} &\geq 
g(\frac{1}{2^{\ell}}) \sum_{i=1}^{2^{D(\ell+1)}} {x_i \choose 2} \nonumber \\& \geq 
g(\frac{1}{2^{\ell}})\big(\frac{n^2}{2^{D(\ell+1) + 1}} - \frac{n}{2}\big).
\end{align}

By Fact \ref{fact.random_alg_yield_rev_third} and equation \ref{equation.metric.2} we are guaranteed that for $c = \frac{2^{D(\ell+1)}}{g(\frac{1}{2^{\ell}})}$, $c \delta \epsilon Rev(O) \geq \delta \epsilon \frac{n}{3} {n \choose 2}$. Combining the above,
\begin{align*}
Rev(A) \geq (1 - \delta - c \delta \epsilon) Rev(O).
\end{align*}

Due to the fact that $g(0) = 1$ and that $g$ is $\ell$-Lipschitz continuous, $g(\frac{1}{2^\ell}) = \Omega(1)$. On the other hand since $D = O(1)$ and $\ell = O(1)$ we may choose $\epsilon$ and $\delta$ small enough in order to guarantee an EPRAS.
\end{proof}

\section{THE DISSIMILARITY CASE}
\label{section.dissimilarity}

\subsection{A Reduction to Constant Sketches}
\label{subsection.dis_offline_reduction_to_constant_trees}

In this section we show how to create a tree that approximates the optimal dissimilarity value. This tree is produced by taking $K(\optdistree)$ for the optimal tree, $\optdistree$ (as defined earlier) and altering it. As opposed to the revenue case, this theorem guarantees $O(1/\epsilon^2)$ internal nodes while maintaining a $(1 - \epsilon)$ approximation. Note that this result holds for any dissimilarity instance and thus may be of independent interest. For an overview we refer the reader to our Techniques section.

\begin{theorem}
\label{theorem.approximate_diss_with_constant_size}
Let $\optdistree$ denote the optimal dissimilarity tree and assume it contains $n$ leaves (i.e., data points). Then, for any $\epsilon>0$, there exists a tree $\constdistree$ such that
(i) $\constdistree$ contains $\Theta(1/\epsilon^2)$ internal nodes, each with at most $3\epsilon^2 n$ children, and (ii) 
$dis(\constdistree) \geq (1 - \epsilon)dis(\optdistree)$.
\end{theorem}

\noindent In order to obtain $\constdistree$ given a binary tree, $T$, we use $K(T)$ (as defined in Section \ref{section.revenue}). We then convert $K(T)$ to $\constdistree$, by randomly partitioning each contracted node's data points into $1/\epsilon$ clusters and attaching them in a ``comb"-like structure. The process is defined in Algorithm \ref{algorithm.convert_T_to_bar(T)} (see Figure \ref{figure.skeleton_to_hc_tree} for an example).\\

\begin{algorithm}[H]
\caption{Algorithm to convert $T$ to $\constdistree$.}
\begin{algorithmic}
\STATE $K(T) \leftarrow$ Algorithm \ref{algorithm.K(T)} applied to $T$.
\FOR{each node $c \in K(T)$ and its data points $D_c$}
	\STATE Partition $D_c$ into $1/\epsilon$ random sets of equal sizes, $P = \{P_1, \ldots, P_{1/\epsilon}\}$.
	\FOR{$P_i \in P$}
	    \STATE Create a new auxiliary node, $u_i$.
	    \STATE Attach $P_i$ as $u_i$'s children.
	    \STATE Create a new node $\ell_i$, and attach it between $c$ and its parent.
	    \STATE Attach $u_i$ as $\ell_i$'s child.
	\ENDFOR
\ENDFOR
\STATE Return the resulting tree as $\constdistree$.
\end{algorithmic}
\label{algorithm.convert_T_to_bar(T)}
\end{algorithm}


\noindent Note that $D_c = \emptyset$ if $c$ is the root (since the root is blue) and therefore $\ell_i$ is indeed only defined for $c$'s that have a parent. Also note that as in Remark \ref{remark.\constrevtree_remains_binary}, $\constdistree$ remains binary if we disregard the auxiliary nodes. Next we show that $\constdistree$ is of constant size and that $|\constdistree_{ij}|$ is (approximately) lower bounded.

\begin{lemma}
\label{lemma.bar(T)_has_constant_size}
$\constdistree$ contains at  most $20/\epsilon^2$ and at least $2/\epsilon^2$ internal nodes with at most $3 \epsilon^2 n$ children.
\end{lemma}

\begin{lemma}
\label{lemma.bar(T)_has_high_dissimilarity}
The resulting tree, $\constdistree$, guarantees in expectation that, $|\constdistree_{ij}| \geq (1-\epsilon)|T_{ij}| - 6 \epsilon n$.
\end{lemma}

\noindent We defer the proofs of Lemmas \ref{lemma.bar(T)_has_constant_size} and \ref{lemma.bar(T)_has_high_dissimilarity} to the Appendix. Finally, combining Lemmas \ref{lemma.bar(T)_has_constant_size} and \ref{lemma.bar(T)_has_high_dissimilarity} for $T = \optdistree$ with Fact \ref{fact.random_alg_dissimilarity}, is enough to prove Theorem \ref{theorem.approximate_diss_with_constant_size}. (For the formal proof, see Appendix).

\subsection{An Efficient-PRAS for Dissimilarity Instances with Not All Small Weights}
\label{subsection.dense_dis_epras}

In this section we consider the problem of finding an optimal dissimilarity tree in instances with weights that are not all small and present an Efficient-PRAS. As in the revenue case, again we show that this is the best one could hope for, and complement our result by showing that the problem is NP-Complete and thus does not admit an optimal, polynomial solution (see Theorem \ref{theorem.dense_diss_NPC} in the Appendix)

Let $\epsilon > 0$ and let $\constdistree_\epsilon$ denote the tree guaranteed by Theorem \ref{theorem.approximate_diss_with_constant_size} for $\epsilon$. As in the revenue case, for an internal node of $\constdistree$, $i$, let $D_i$ denote the set of data points that are $i$'s children and let $W_{ij}$ denote the set of (dissimilarity) edges crossing between $D_i$ and $D_j$. Therefore, $dis(\constdistree_\epsilon) = \sum_{i,j \in S}\big( W_{ij} \sum_{\ell \in S}  |D_\ell | \big) + b$, where the second sum is over all sets $D_\ell$ contained in $\constdistree_{ij}$ (as defined by $\constdistree_\epsilon$'s sketch). Furthermore, $b$ is defined as the dissimilarity gained by nodes within the same "star" structure. Theorem \ref{theorem.approximate_diss_with_constant_size} guarantees that $|D_i|$ is small - therefore, since our instance has weights that are not all small (and by Fact \ref{fact.random_alg_dissimilarity} the optimal solution is large) this dissimilarity is negligible and we may assume $b=0$ since we already lose a factor of $1 - \epsilon$. Finally, recall that $|S| \leq 20k$.


Our Efficient-PRAS follows as in the revenue case and is therefore deferred to the Appendix (Algorithm \ref{algorithm.guess_OPT_eps_in_dense_dis_case}). The following theorem is proven identically to the revenue case and is therefore omitted.

\begin{theorem}
\label{theorem.dense_dis_epras.epras}
Algorithm \ref{algorithm.guess_OPT_eps_in_dense_dis_case} is an EPRAS for dissimilarity instances with weights that are not all small.
\end{theorem}

\section{HARDNESS RESULTS FOR INSTANCES WITH NOT ALL SMALL WEIGHTS}
\label{section.hardness_for_dense_cases}
When considering instances with weights that are not all small, we have only shown Efficient-PRAS's up until now. To complement our results, we show that we can not hope for optimal, polynomial algorithms, assuming the Small Set Expansion (SSE) hypothesis. (For a formal definition of SSE see \cite{Approximate_Hierarchical_Clustering_via_Sparsest_Cut_and_Spreading_Metrics}). In fact, it is enough to show that these objectives are NP-complete assuming the instances are (1) unweighted and (2) guarantee that $\sum_{i<j}w_{ij} = \Omega(n^2)$. We call such instances \emph{dense} instances.

\begin{theorem}
\label{theorem.dense_rev_NPC}
The Revenue objective for dense instances is in NPC (assuming SSE).
\end{theorem}

\begin{theorem}
\label{theorem.dense_diss_NPC}
The Dissimilarity objective for dense instances is in NPC (assuming SSE).
\end{theorem}

\begin{theorem}
\label{theorem.dense_hcc_NPC}
The $\texttt{HCC}^{\pm}$ objective is in NPC (assuming SSE).
\end{theorem}
\section{HIERARCHICAL CORRELATION CLUSTERING}
\label{section.hcc}

In this section we consider the case where the collected data may contain both similarity and dissimilarity information. We first show a worst case approximation and thereafter show an Efficient-PRAS for $\texttt{HCC}^{\pm}$.


\subsection{Worst Case Guarantees for HCC}


Here we consider two separate algorithms which, if combined properly, will yield our approximation. The first is a simple greedy algorithm whereas the second optimizes for the \textsc{Max-Uncut Bisection} problem for its top most cut and then continues with the greedy algorithm. We first show baseline guarantees of the greedy algorithm and then use the work of \citet{Hierarchical_Clustering:_a_0.585_Revenue_Approximation} in order to obtain guarantees on the second algorithm with respect to the \ref{equation.objective_hcc} objective. We defer the following proof to the appendix.

\begin{proposition}
\label{proposition.hcc_greedy_alg}
There exists a greedy algorithm, denoted by $ALG_{GRE}$, that returns an HC tree $T_1$ guaranteeing, 
\[
hcc(T_1)\ge \tfrac13(n-2)\sum_{ij} w^s_{ij} + \tfrac23n\sum_{ij} w^d_{ij}.
\]
\end{proposition}

Denote by $ALG_{MUB}$ the algorithm that generates an HC tree by first cutting according to \textsc{Max-Uncut Bisection} based on the similarity weights of the instance and then running $ALG_{GRE}$ on each of the two resulting sides. Let $\texttt{OPT}=\texttt{OPT}_s + \texttt{OPT}_d$ be the value of the optimum \ref{equation.objective_hcc} tree where $\texttt{OPT}_s = \sum w^s_{ij}(n - |O_{ij}|)$ and  $\texttt{OPT}_d = \sum w^d_{ij}|O_{ij}|$, defined such that $O_{ij}$ denotes the number of leaves in the subtree rooted at the LCA of $i$ and $j$ in the tree of $\texttt{OPT}$.

\begin{lemma}
\label{lemma.hcc_mub_alg}
Let $T_2$ denote the HC tree returned by $ALG_{MUB}$. Therefore,
\[
hcc_{G}(T_2)\ge 0.585\cdot \texttt{OPT}_s  + \tfrac13 \cdot \texttt{OPT}_d 
\]
\end{lemma}

\begin{proof}
For ease of exposition let $T_2 = T$. The top-split of $T$ is a bisection which means that $|L|=|R|=\tfrac n2$. For ease of notation let: 
\[
W_L^s = \sum_{i,j\in L}w^s_{ij} \text{ and }W_L^d = \sum_{i,j\in L}w^d_{ij}
\]
Similarly, we define $W_R^s$ and $W_R^d$.
Notice that for the $L$ side, \textsc{Greedy} will contribute at least $\tfrac 23 \cdot \tfrac n2\cdot W_L^d$ to $\sum w^d_{ij}|T_{ij}|$, as per Proposition \ref{proposition.hcc_greedy_alg}. Similarly, for the $R$ side. This means that in the tree $T$, any edge contributes either $\tfrac23\cdot \tfrac n2$ (if it was cut by \textsc{Greedy}) or $n$ (if it was cut at the top-split of \textsc{Max-Uncut Bisection}). In any case, we have: 
\begin{equation}\label{eq:lem1}
   \sum w^d_{ij}|T_{ij}|\ge \tfrac23 \cdot \tfrac n2 \sum w^d_{ij} \ge\tfrac13\texttt{OPT}_d
\end{equation}
by using the upper bound $\texttt{OPT}_d\le n\sum w^d_{ij}$.

We now deal with $\texttt{OPT}_s$. Observe that:
\begin{align*}
\sum w^d_{ij}(n-|T_{ij}|) &\ge 
W_L^+ (\tfrac n2+ \tfrac13\tfrac n2)  + W_R^s(\tfrac n2+ \tfrac13\tfrac n2) \\ &\geq
\tfrac23n(W_L^s + W_R^s)    
\end{align*}
since every edge within $L$ will contribute $\tfrac n2$ due to the bisection, plus an extra $\tfrac13\tfrac n2$ due to the greedy step. The same is true for edges in $R$. 

Finally, since we used a 0.8776 for \textsc{Max-Uncut Bisection}, it holds directly from~\cite{Hierarchical_Clustering:_a_0.585_Revenue_Approximation} that:
\begin{equation}\label{eq:lem2}
    \sum w^d_{ij}(n-|T_{ij}|)\ge \tfrac23 \cdot 0.8776 \cdot\texttt{OPT}_s \ge 0.585\cdot \texttt{OPT}_d
\end{equation}

The lemma follows by summing eq.~(\ref{eq:lem1}) and~(\ref{eq:lem2}).
\end{proof}

Finally, we combine Proposition \ref{proposition.hcc_greedy_alg} and Lemma \ref{lemma.hcc_mub_alg} in order to yield the following Theorem (whose proof is defered to the appendix).

\begin{theorem}
\label{theorem.hcc.worst_case}
Running $ALG_{GRE}$ with probability $p$ and otherwise $ALG_{MUB}$ guarantees an approximation of 0.4767 for the \ref{equation.objective_hcc} objective, when p = 0.43.
\end{theorem}

\subsection{An Efficient-PRAS for HCC on complete graphs}

Here we consider the $\texttt{HCC}^{\pm}$ objective (as defined earlier in the introduction) and show an Efficient-PRAS. We also complement our results and show that in fact this problem is NP-Complete and thus we cannot hope for an optimal, polynomial solution (see Theorem \ref{theorem.dense_hcc_NPC} in the Appendix).

Let $\texttt{ALG}^{\pm}$ denote the algorithm that runs Algorithm \ref{algorithm.guess_OPT_eps_in_dense_rev_case} and Algorithm \ref{algorithm.guess_OPT_eps_in_dense_dis_case} simultaneously and returns the tree maximizing the $\texttt{HCC}^{\pm}$ objective. We prove that $\texttt{ALG}^{\pm}$ is in fact an Efficient-PRAS for the $\texttt{HCC}^{\pm}$ objective. We defer the theorem's proof to the appendix.

\begin{theorem}
\label{theorem.hcc.dense_epras}
$\texttt{ALG}^{\pm}$ is an Efficient-PRAS for the $\texttt{HCC}^{\pm}$ objective.
\end{theorem}

\section{CONCLUSION}


In this paper we show that to optimize for the \ref{equation.objective_rev} and \ref{equation.objective_dis} objectives, it suffices to consider HC trees with constant-sized sketches, thereby greatly simplifying these problems. This result can be applied to both the heuristic setting (since it greatly reduces the range of optimal solutions that need to be considered) and the approximation setting. Specifically, an approximation algorithm may iterate over all constant sized trees. Thereafter, it will need to partition the data points into the leaves of the constant-sized tree - thus reducing our problem to the well-studied realm of graph partitioning problems.

We then consider the family of instances with weights that are not all small. We show Efficient-PRAS's for both \ref{equation.objective_rev} and \ref{equation.objective_dis} objectives. Furthermore, we show that this family of instances encompasses many metric-based similarity instances. Finally, we introduce the $\ref{equation.objective_hcc}$ objective which we hope will provide a better connection between the realms of correlation and hierarchical clustering. We then show a worst case approximation of 0.4767 and show an Efficient-PRAS for the $\texttt{HCC}^{\pm}$ objective that leverages our algorithms presented for the \ref{equation.objective_rev} and \ref{equation.objective_dis} objectives for instances with weights that are not all small. 
\section{ACKNOWLEDGEMENTS}

The authors would like to deeply thank Claudio Gentile and Fabio Vitale for their helpful discussions and insights regarding the connection to metric-based similarity instances. We also thank Sara Ahmadian and Alessandro Epasto for interesting discussions during early stages of our work.
\bibliographystyle{plainnat}
\bibliography{bib.bib}
\clearpage

\appendix

\section{DEFERRED PROOFS OF SUBSECTION \ref{subsection.rev_sketch_reduction}}

\begin{proof}[Proof of Lemma \ref{lemma.dasguptas_lemma}]
We first note that the removal of any edge creates two binary trees. Next we show how to find an edge satisfying the rest of the properties.

Given the rooted tree $T$, we travel down the tree from the root such that we always pick the child that contains more data points in its subtree (compared to the other child, if another child exists). We denote the $i$'th node along this path that contains exactly two children, by $u_i$ for $i \in \{1,2,\ldots\}$. Furthermore, we denote the sets of data points contained by its two children by $A_i$ and $B_i$ such that, $|A_i| \geq |B_i|$.

Let $k^* := \arg \min_i \{|B_1| + \cdots + |B_i| \geq \frac{n}{3} \}$. Since $|A_{k^*}| + |B_1| + \cdots |B_{k^*}| = n$, we are guaranteed that $|A_{k^*}| \leq \frac{2n}{3}$. On the other hand, since $|A_{k^*}| \geq |B_{k^*}|$ and $|A_{k^*}| + |B_{k^*}| = n - (|B_1| + \cdots |B_{k^*-1}| )$ we are also guaranteed that, $|A_{k^*}| \geq \frac{n}{3}$. 

Therefore, removing the edge between $u_{k^*}$ and its child associated with $A_{k^*}$ guarantees that the resulting trees each have at most at least $n/3$ data points thereby completing the proof.
\end{proof}

\section{DEFERRED PROOFS AND DEFINITIONS OF SUBSECTION \ref{subsection.dense_rev}}

\begin{observation}
\label{observation.dense_graph_opt}
Due to Fact \ref{fact.random_alg_yield_rev_third} if we denote by $\optrevtree$ our optimal solution, then since our instance is $\rho, \tau$-weighted we get,
\[ 
rev(\optrevtree) \geq \frac{\rho \tau n^3}{3},
\] 
for some smaller, yet still constants $\rho$ and $\tau$.
\end{observation}

\begin{proof}[Proof of Lemma \ref{lemma.algorithms_approximation}]
Let $T_{alg}$ denote the tree returned by Algorithm \ref{algorithm.guess_OPT_eps_in_dense_rev_case}. Furthermore denote by $\alpha_\ell$ and $\beta_{ij}$ the real values of $\constrevtree_\epsilon$. Therefore,
\begin{align*}
rev(T_{alg}) &\geq 
\sum_{i \leq j} \sum_{\ell \in S} \big( (\alpha_\ell - n \epsilon^2 - n \epsilon_{err}) \\ &\cdot
(\beta_{ij} - n^2 \epsilon^3 - n^2 \epsilon_{err}) \big) \\ & \geq 
\sum_{i \leq j} \sum_{\ell \in S} \big( \alpha_\ell \beta_{ij} \big) - 
\sum_{i \leq j} \sum_{\ell \in S} \big( \beta_{ij} n\epsilon^2 \big) \\ &- 
\sum_{i \leq j} \sum_{\ell \in S} \big( \beta_{ij} n\epsilon_{err} \big) - 
\sum_{i \leq j} \sum_{\ell \in S} \big( \alpha_\ell n^2\epsilon^3\big) \\ &- 
\sum_{i \leq j} \sum_{\ell \in S} \big( \alpha_\ell n^2\epsilon_{err}\big) \\ & \geq
\big( \sum_{i \leq j} \sum_{\ell \in S}   \alpha_\ell \beta_{ij} \big) - n^3 \epsilon^2 20k - n^3 \epsilon_{err} 20k \\ &- 
n^3 \epsilon^3 (20k)^2 - n^3 \epsilon_{err} (20k)^2 \\ &=
\big( \sum_{i \leq j} \sum_{\ell \in S} \alpha_\ell \beta_{ij} \big) - n^3\big( \epsilon^2 20k \\ &+ 
\epsilon^3 (20k)^2 +  \epsilon_{err} 20k  + \epsilon_{err} (20k)^2 \big)\\ &\geq
rev(\constrevtree_\epsilon) - n^3 (421\epsilon + 20k\epsilon_{err} +  400k^2 \epsilon_{err}),
\end{align*}
where the first inequality follows from the property tester's guarantees and the fact that we did not guess $\alpha_\ell$ and $\beta_{ij}$ to their exact values. The third inequality follows since there are at most $k$ sets in the partition, $\sum \beta_{ij} \leq n^2$ and $\sum \alpha_\ell \leq n$. The last inequality is due to the fact that $k \leq 1/\epsilon + 1$ and $\epsilon$ is chosen to be small enough.

Due to Observation \ref{observation.dense_graph_opt}, Theorem \ref{theorem.approximate_rev_with_constant_size} and by choosing $\epsilon_{err} = \frac{\epsilon^3}{400}$, we get,
\begin{align*}
rev(T_{alg}) &\geq 
rev(\constrevtree_\epsilon) - n^3(O(\epsilon)) \\ &\geq
rev(\constrevtree_\epsilon) - \frac{O(\epsilon)}{\rho \tau} rev(\optrevtree) \\& \geq 
(1 - O(\epsilon) - \frac{O(\epsilon)}{\rho \tau})rev(\optrevtree).
\end{align*}
Thus by choosing $\epsilon$ small enough, we get the desired result.
\end{proof}

\section{DEFERRED PROOFS OF SUBSECTION \ref{subsection.dis_offline_reduction_to_constant_trees}}

\begin{proof}[Proof of Lemma \ref{lemma.bar(T)_has_constant_size}]
Consider the proof of Lemma \ref{lemma.\constrevtree_has_constant_size}. The only difference between $\constrevtree$ and $\constdistree$ (with respect to the number of their internal nodes) is the fact that in $\constdistree$ the contracted nodes are multiplied by $1/\epsilon$ (and therefore the auxiliary nodes as well). Thus, clearly the lemma holds.
\end{proof}

\begin{proof}[Proof of Lemma \ref{lemma.bar(T)_has_high_dissimilarity}]
In order to prove the lemma we consider the following observations. The first of which is Observation \ref{observation.unchanged_children_sets} which holds here as well. The second is the following.

\begin{observation}
\label{observation.pair_in_same_comb_replacement}
Consider any two data points, $i$ and $j$, that are contained in the same contracted node in $K(\optdistree)$. Further assume that they end up under different auxiliary nodes. Therefore, any descendant of the corresponding contracted node (in $K(\optdistree)$) is contained in $\constdistree_{ij}$.
\end{observation}

Consider two data points in $\optdistree$, $i$ and $j$ and consider some $k \in \optdistree_{ij}$. As before, we denote their lca's by $v_{ik}$, $v_{jk}$ and $v_{ij}$ and assume without loss of generality that $i$ is clustered first with $k$ and therefore, $v_{ij} = v_{kj}$.

We would like to bound the number of $k$'s for which $k \not \in \constdistree_{ij}$. As before, let $\{T_\ell^{B \cup G}\}$ denote the set of trees defined by $\optdistree - (B \cup G)$ and let $T_i^{B\cup G}$ (resp. $T_j^{B\cup G}$ and $T_k^{B\cup G}$) denote the tree in $\optdistree - (B \cup G)$ containing $i$ (resp. $j$ and $k$). If $k \in T_i^{B\cup G}$ or $k \in T_j^{B\cup G}$ then since the number of data points contained in these trees is at most $6 \epsilon n$, we may disregard such $k$'s and incur an additive loss of $6 \epsilon n$. Therefore, we assume, $k \not \in T_i^{B\cup G}$ and $k \not \in T_j^{B\cup G}$.

Thus, we split into the following cases. The first is the case where $v_{jk}$ is green/blue. Otherwise, this means that $v_{jk}$ has at most one child with a blue descendant. It can not be the child containing $j$ since that would mean that $k \in T_i^{B\cup G}$. Thus, we may only consider the following final cases: either exists a green/blue node on the path $v_{ik} \rightarrow v_{ij}$ or there must exist a green/blue node both on the path $k \rightarrow v_{ik}$ and on the path $i \rightarrow v_{ik}$ (since $k \not \in T_i^{B\cup G}$). Otherwise, exists a green/blue node on the path $k \rightarrow v_{ik}$ and not on the path $i \rightarrow j$.

We prove our lemma for each of these cases.

\begin{enumerate}
\item $v_{jk}$ is green/blue: Due to Observation \ref{observation.unchanged_children_sets} we are guaranteed that $k \in \constdistree_{ij}$.
\item There exists a green/blue node on the path $v_{ik} \rightarrow v_{ij}$: Due to Observation \ref{observation.unchanged_children_sets} we are guaranteed that $k \in \constdistree_{ij}$.
\item There exists a green/blue node both on the path $k \rightarrow v_{ik}$ and on the path $i \rightarrow v_{ik}$: In this case $v_{ik}$ is green/blue and therefore, again due to Observation \ref{observation.unchanged_children_sets} we are guaranteed that $k \in \constdistree_{ij}$.
\item There exists a green/blue node on the path $k \rightarrow v_{ik}$ and not on the path $i \rightarrow j$: In this case $i$ and $j$ are in the same contracted node in $K(\optdistree)$. If they end up under different auxiliary nodes, then by Observation \ref{observation.pair_under_same_auxiliary_parent} $k \in \constdistree_{ij}$. Since we partitioned the data points in the contracted nodes randomly (under restriction that the sets are of the same size), the probability that $i$ and $j$ will end up under different auxiliary nodes is $\geq (1-\epsilon)$. 
\end{enumerate}

Thus, in any case, $E[|\constdistree_{ij}|] \geq (1-\epsilon)|\optdistree_{ij}| - 6\epsilon n$.
\end{proof}

\begin{proof}[Proof of Theorem \ref{theorem.approximate_diss_with_constant_size}]
Lemma \ref{lemma.bar(T)_has_constant_size} guarantees the first bullet. For the second bullet, denote by $\optdistree$ the optimal solution. We note that $\optdistree$ is binary. Furthermore, due to Lemma \ref{lemma.bar(T)_has_high_dissimilarity} and Fact \ref{fact.random_alg_dissimilarity}, we get,
\begin{align*}
E[dis(\constdistree)] &= 
\sum_{i < j} w_{ij} E[|\constdistree_{ij}|] \\ &\geq
\sum_{i < j} w_{ij}((1-\epsilon)|\optdistree_{ij}| - 12 \epsilon n) \\ &=
(1-\epsilon)dis(\optdistree) - 12 \epsilon n \sum_{i < j} w_{ij} \\ &\geq
(1 - 38 \epsilon)dis(\optdistree).
\end{align*}
Since the expectation is over trees with our desired characteristics (i.e., constant number of internal nodes and each node contains a small number of children), we deterministically take $\constdistree$ to be the tree maximizing the expectation. Thus, by choosing $\epsilon' = \epsilon/38$ we get the desired result.
\end{proof}

\section{DEFERRED ALGORITHMS OF SUBSECTION \ref{subsection.dense_dis_epras}}

\begin{algorithm}[H]
\caption{EPRAS for the dense dissimilarity case.}
\begin{algorithmic}
\STATE Enumerate over all trees, $T$, with $k$ internal leaves.
\FOR{each such $T$}
	\FOR{$\{\alpha_i\}_{i\leq k} \subset \{i \epsilon^2 n: i \in \field{N} \land  i \leq \frac{3}{\epsilon} \}$}
		\FOR{$\{\beta_{ij}\}_{i\leq k,j\leq k} \subset \{i \epsilon^3 n^2: i \in \field{N} \land  i \leq \frac{9}{\epsilon} \}$}
			\STATE Run $PT(\{\alpha_i\}, \{\beta_{ij}\}, \epsilon_{err}=\epsilon^3, \delta)$.
		\ENDFOR
	\ENDFOR
	\STATE Compute the dissimilarity based on $T$ and $PT$'s output.
\ENDFOR
\STATE Return the maximal dissimilarity tree encountered.
\label{algorithm.guess_OPT_eps_in_dense_dis_case}
\end{algorithmic}
\end{algorithm}

\section{DEFERRED PROOFS OF SECTION \ref{section.hcc}}

\begin{proof}[Proof of Proposition \ref{proposition.hcc_greedy_alg}]
For each vertex $v\in V$, our algorithm maintains scores $s(v)$ which are initially set to zero. The algorithm will actually remove the node of largest score at each step and recurse on the remaining vertices, hence producing a caterpillar tree (a tree whose every internal node has at least one leaf). A similar greedy strategy to the one described below can also produce a tree (not necessarily caterpillar) in a bottom-up fashion by repeatedly merging node pairs. Notice that the algorithm is deterministic. 

For every edge $(i,j)$ of similarity weight $w_{ij}^s$,
decrease $s(i)$ and $s(j)$ by $\tfrac{n-2}{2}w_{ij}^s$, and increase every other score $s(k)$ by $w_{ij}^s$, where $k\in V\setminus \{i,j\}$. The intuition behind such assignments, is that for a pair $i,j$ of similarity $w_{ij}^s$, whenever we remove another node $k$ first, $k$'s contribution to the $hcc$ objective increases by $w_{ij}^s$, as $k$ lies outside of the lowest common ancestor between $i,j$. Similarly, for every edge $(i,j)$ of dissimilarity $w_{ij}^d$, we increase $s(i)$ and $s(j)$ by $\tfrac{n}{2}w_{ij}^d$, and decrease every other score $s(k)$ by $w_{ij}^d$, where $k\in V\setminus \{i,j\}$. 

Next, let $u\in V$ have the largest score and $V'=V\setminus \{u\}$. Remove $u$ and any adjacent edges from the graph, then recursively construct a tree $T_1'$ restricted on $V'$ for its leaves (if $|V'|=2$, just output the unique binary tree on the two nodes). The final output of the algorithm is a new tree $T_1$ with one child being $u$ and the other child being the root of $T_1'$. 

We now prove correctness: Let $u$ as above and let $w_u^s= \sum_{(u,v)}w_{uv}^s, w_u^d=\sum_{(u,v)}w_{uv}^d, W^s=\sum_{(i,j)}w_{ij}^s, W^d =\sum_{(i,j)}w_{ij}^d$. Notice that according to the scoring rule of our algorithm: 
\[s(u)=(W^s-w_u^s) - \tfrac{n-2}{2}w_u^s -(W^d-w_u^d) + \tfrac{n}{2}w_u^d\]

Note that by induction, tree $T_1'$ that has $n-1$ leaves, satisfies the conclusion of the proposition:
\begin{equation}\label{eq:prop1}
    hcc(T_1')\ge\tfrac13(n-3)(W^s-w_u^s)+\tfrac23(n-1)(W^d-w_u^d)
\end{equation}

Since $u$ had the largest score, it follows that $s(u)\ge0$. Therefore:
\[
(W^s-w_u^s)-(W^d-w_u^d) \ge\tfrac{n-2}{2}w_u^s + \tfrac{n}{2}w_u^d
\]
We add $\tfrac12[(W^s-w_u^s)-(W^d-w_u^d)]$ to both sides:
\[
(W^s-w_u^s)-(W^d-w_u^d) \ge\tfrac{1}{3}(hcc_u^s-(W^d-w_u^d)-nw_u^d)
\]
where $hcc_u^s=(n-2)w_u^s+(W^s-w_u^s)$ is the total contribution $u$ can have due to similarity weights in any tree. By rearranging terms:
\begin{equation}\label{eq:prop2}
(W^s-w_u^s)+nw_u^d \ge\tfrac{1}{3}hcc_u^s+\tfrac23hcc_u^d    
\end{equation}
where $hcc_u^d=(W^d-w_u^d)+nw_u^d$ is the total contribution $u$ can have due to dissimilarity edges in any tree. 

Let $hcc_u(T_1)$ be the contribution towards the $hcc$ objective of node $u$ in $T_1$ and observe we can easily compute this quantity as $u$ got removed first. In other words, $hcc_u(T_1)= (W^s-w_u^s)+nw_u^d$, as any dissimilarity edge $(u,\cdot)$ has a lowest common ancestor of size $n$ and for every similarity edge $(i,j), i,j\neq u$, $u$ is a non-leaf of $T_{ij}$. Summing up eq.~(\ref{eq:prop1}) and~(\ref{eq:prop2}), and noting that $hcc(T_1)=hcc_u(T_1)+hcc(T_1')$ concludes the proof.
\end{proof}

\begin{proof}[Proof of Theorem \ref{theorem.hcc.worst_case}]
A simple calculation suggests that the expected value for \texttt{HCC} is at least:
\[
\min_{p}\left\{ p\cdot \tfrac13 + 0.585\cdot(1-p),p\cdot\tfrac23+(1-p)\cdot\tfrac{1}{3}\right\}
\]
By balancing the two terms, the minimum is achieved when the parameter $p=1-\tfrac{\tfrac13}{0.585}$ and the final approximation factor becomes 0.4767.
\end{proof}

\begin{proof}[Proof of Theorem \ref{theorem.hcc.dense_epras}]
There are two cases to consider: either $\sum_e w_e^d \geq \sum_e w_e^s$ or $\sum_e w_e^d \leq \sum_e w_e^s$. We first consider the case that $\sum_e w_e^d \geq \sum_e w_e^s$ (the second is handled symmetrically). We rewrite the objective function for some HC tree $T$.
\begin{align*}
\label{equation.unweighted_dense_hcc.hcc_written_as_diss}
hcc^{\pm}(T) \nonumber &=
\sum_e w_e^d(T_e) + \sum_e w_e^s(n - T_e) \\ &=
\sum_e w_e^d(T_e) + \sum_e (1 - w_e^d)(n - T_e) \\ &=
2\sum_e w_e^d(T_e) + \sum_e (n - T_e) - n \sum_e w_e^d \\ &=
2\sum_e w_e^d(T_e) + \frac13 n {n \choose 2} - n \sum_e w_e^d,
\end{align*}
where the last equality follows from Fact \ref{fact.revenue_of_clique}. We first observe that a tree that maximizes the dissimilarity instance defined by $w_e^d$ is a tree that maximizes the original  $\texttt{HCC}^{\pm}$ objective. Let $O^d$ denote the tree maximizing the dissimilarity objective and let $O$ denote the tree maximizing the $\texttt{HCC}^{\pm}$ objective. By Theorem \ref{theorem.dense_dis_epras.epras} we know that for any constant $\epsilon >0$ algorithm \ref{algorithm.guess_OPT_eps_in_dense_dis_case} (denoted henceforth as $ALG$) generates dissimilarity of at least $(1 - \epsilon)\sum_e w_e^d(O^d_e) = (1 - \epsilon)\sum_e w_e^d(O_e)$. Therefore, for any $\epsilon > 0$,
\begin{align*}
hcc^{\pm}(ALG) &= 
2\sum_e w_e^d(ALG_e) + \frac13 n {n \choose 2} - n \sum_e w_e^d \\ &\geq 
2(1 - \epsilon)\sum_e w_e^d(O_e) \\&+ 
\frac13 n {n \choose 2} - n \sum_e w_e^d \\ &=
(1 - 2\epsilon)\sum_e w_e^d(O_e) \\&+ 
\sum_e w_e^d(O_e) + \frac13 n {n \choose 2} - n \sum_e w_e^d \\ &\geq
(1 - 2\epsilon) \\& \cdot 
\big( \sum_e w_e^d(O_e) + \frac13 n {n \choose 2} - n \sum_e w_e^d\big) \\ &=
hcc^{\pm}(O),
\end{align*}
where the last inequality follows from Fact \ref{fact.random_alg_dissimilarity}. 

The case that $\sum_e w_e^d \leq \sum_e w_e^s$ is solved symmetrically (using Theorem \ref{theorem.dense_rev_epras.epras} and Fact \ref{fact.random_alg_yield_rev_third}) which concludes the proof.
\end{proof}

\section{HARDNESS RESULTS}

\begin{proof}[Proof of Theorem \ref{theorem.dense_rev_NPC}]
Note that clearly the problem is in NP (since given a tree its revenue may be checked efficiently), therefore we only need to show that it is NP-hard.

\citet{Bisect_and_Conquer:_Hierarchical_Clustering_via_Max-Uncut_Bisection} showed that the unweighted revenue case is APX-hard under the Small Set Expansion hypothesis. This in turn guarantees that the unweighted revenue problem is NP-hard assuming the Small Set Expansion. Next we show how to reduce an unweighted revenue instance to a dense unweighted revenue instance (in polynomial time).

Roughly speaking we will simply add a disconnect clique of size $n$ to the general graph. Formally, let $G = (D,E_D,w)$ denote a general revenue instance such that, $D = \{d_1, \ldots, d_n \}$. We convert $G$ to a dense instance $G' = (V, E_V, w')$ simply by adding a clique of size $n$ (disconnected from $V$) with similarities of size 1. We denote this clique's set of nodes by $L = \{ \ell_1, \ldots, \ell_n\}$. Therefore, $w'(\ell_i, \ell_j) = 1, w'(d_i, d_j) = w(d_i, d_j)$ and $w'(\ell_i, d_j) = 0$. 

Clearly $G'$ is dense. Let $T'$ denote the optimal solution to $G'$. It is known that the optimal tree first cuts the disconnected components of $G'$. Therefore, there exists a node $u$ in $T'$ such that the subtree rooted at $u$ contains the entirety of $L$ and no data points from $D$. Since $D$ is disconnected from $L$ and due to the definition of the revenue goal function, taking $u$ and moving it to the top of $T'$ (formally, if $r'$ is the root of $T'$, then we create a new root, $r$ and attach $u$ and $r'$ as its immediate children), can only increase $T'$'s revenue. Thus, we may assume w.l.o.g. that in $T'$ the root already disconnects $L$ and $D$.

Let $v_D$ and $v_L$ denote $T'$'s root's immediate children containing $D$ and $L$ respectively. Let $T'_D$ denote the subtree rooted at $u_D$. $T'_D$ is clearly optimal for instance $G$ (since otherwise, we could have replaced $T'_D$ with the optimal tree for $G$, thereby increasing $T'$'s revenue, contradicting the fact that it is optimal).

Thus, we converted, in polynomial time, the optimal tree for $G'$ to the optimal tree for $G$, proving that the dense revenue problem is NP-hard.
\end{proof}

\begin{definition}
We say that an unweighted graph is complement-dense if its complement graph (i.e., the graph we get by removing all existing edges and adding all missing edges) is dense.
\end{definition}

\begin{lemma}
\label{lemma.sparse_rev_NPC}
The problem of finding a maximal revenue tree for revenue instances which are complement-dense is NP-complete (assuming the Small Set Expansion hypothesis).
\end{lemma}

\begin{proof}
Note that clearly the problem is in NP (since given a tree its revenue may be checked efficiently),therefore we only need to show that it is NP-hard.

As in Theorem \ref{theorem.dense_rev_NPC}, we reduce an unweighted revenue instance to a complement-dense unweighted revenue instance. Specifically we do this by adding a disconnected path of length $n^2$ to the original graph. Formally, let $G = (D,E_D,w)$ denote a general revenue instance such that, $D = \{d_1, \ldots, d_n \}$. We convert $G$ to a complement-dense instance $G' = (V, E_V, w')$ simply by adding a path of size $n^2$ (disconnected from $V$) with similarities of size 1. We denote this path's set of nodes by $L = \{ \ell_1, \ldots, \ell_{n^2}\}$. Therefore, $w'(\ell_i, \ell_{i+1}) = 1, w'(d_i, d_j) = w(d_i, d_j)$ and $w'(\ell_i, d_j) = 0$. Note that $G'$ is clearly complement-dense.

As in the proof of Theorem \ref{theorem.dense_rev_NPC} exists a node $u$ in the optimal solution of $G'$, $T'$, such that $u$ contains the entirety of $L$ and no data points from $D$. Again, we may move $u$ and its subtree to the root of $T'$ thereby only increasing the revenue. Thus, given $T'$ we may take its child that contains $D$ as our optimal tree for $G$.
\end{proof}

\begin{observation}
Since the problem of finding a minimal (Dasgupta) cost tree is the dual problem of the revenue problem, the unweighted, complement-dense Dasgupta cost problem is NP-complete (assuming the Small Set Expansion hypothesis).
\end{observation}

\begin{proof}[Proof of Theorem \ref{theorem.dense_diss_NPC}]
Note that clearly the problem is in NP (since given a tree its dissimilarity may be checked efficiently), therefore we only need to show that it is NP-hard. We do this by reducing the unweighted, complement-dense Dasgutpa cost problem to this problem.

Roughly speaking we simply consider the complement graph of the HC instance. Formally, given a complement-dense HC instance $G = (V,E,w)$ we define its complement as $G_c = (V_c, E_c, w_c)$. Therefore, for any edge $e$, $w_c(e) = 1 - w(e)$. Thus,
\begin{align*}
\min_T cost_G(T) &= 
\min_T \sum w(e) |T_e| \\ &= 
\min_T \sum (1 - w_c(e))|T_e|.
\end{align*}
\cite{a_cost_function_for_similarity-based_hierarchical_clustering} proved that for any binary tree $T$ and for any HC instance which is a clique $H$ its cost is fixed and $cost_H(T) = \frac{1}{3}(|V(H)|^3 - |V(H)|)$. Since the optimal tree for this cost function is in fact binary we get,
\begin{gather*}
\min_T \sum (1 - w_c(e))|T_e| = \\
\frac{1}{3}(|V(G)|^3 - |V(G)|) - \max_T \sum w_c(e)|T_e|.
\end{gather*}

Since $w$ defines a complement-dense instance, $w_c$ defines a dense instance. Thus, we reduced our original problem to $\max_T \sum w_c(e)|T_e|$ such that $w_c$ is dense, thereby completing the proof.
\end{proof}

\begin{proof}[Proof of Theorem \ref{theorem.dense_hcc_NPC}]
The theorem is proven simply by rewriting the $\texttt{HCC}^{\pm}$ objective in terms of either revenue or dissimilarity (choosing that which contributes more to the total weight) as in the proof of Theorem \ref{theorem.hcc.dense_epras} and then using Theorems \ref{theorem.dense_rev_NPC} and \ref{theorem.dense_diss_NPC}.
\end{proof}

\end{document}